\documentclass[longbibliography, aps, pra, amsmath, amssymb, 
amsfonts, twocolumn, superscriptaddress, footinbib]{revtex4-2}

\pdfoutput=1

\usepackage{graphicx}
\usepackage{subfigure}
\usepackage{dcolumn}
\usepackage{bm}
\usepackage{bbm}
\usepackage{amsthm}
\usepackage{thmtools}
\declaretheorem{theorem}
\usepackage{mathtools}
\usepackage{physics}
\usepackage{binarytree}
\usepackage{tikz,pgfplots}
\usepackage{verbatim}
\usepackage{pgfplots}
\usepackage{placeins}
\usepackage{algorithm2e}
\usetikzlibrary{decorations.pathreplacing}
\usetikzlibrary{arrows.meta}
\usetikzlibrary{graphs}
\usepackage{xcolor}
\usepackage{rotating}
\usepackage{yquant}
\usepackage{tabularray}
\usepackage{array}
\usepackage[colorlinks, linkcolor=red, anchorcolor=blue, citecolor=green]{hyperref}
\RestyleAlgo{ruled}
\pgfplotsset{compat=1.16}

\DeclarePairedDelimiter\intervaloc{\lparen}{\rbrack}

\begin{document}

\title{Measuring Incompatible Observables with Quantum Neural Networks}

\author{Muchun Yang}
\affiliation{Institute of Physics, Beijing National Laboratory for Condensed
  Matter Physics,\\Chinese Academy of Sciences, Beijing 100190, China}
\affiliation{School of Physical Sciences, University of Chinese Academy of
  Sciences, Beijing 100049, China}

\author{Yibin Huang}
\affiliation{Institute of Physics, Beijing National Laboratory for Condensed
  Matter Physics,\\Chinese Academy of Sciences, Beijing 100190, China}
\affiliation{School of Physical Sciences, University of Chinese Academy of
  Sciences, Beijing 100049, China}
  
\author{D. L. Zhou} \email[]{zhoudl72@iphy.ac.cn}
\affiliation{Institute of Physics, Beijing National Laboratory for Condensed
  Matter Physics,\\Chinese Academy of Sciences, Beijing 100190, China}
\affiliation{School of Physical Sciences, University of Chinese Academy of
  Sciences, Beijing 100049, China}

\date{\today}

\begin{abstract}
  The Heisenberg uncertainty principle imposes a fundamental restriction in quantum mechanics, stipulating that measuring one observable completely erases the information on its conjugate one, thereby preventing simultaneous measurements of incompatible observables. Quantum neural networks (QNNs) is one of the most significant applications on near-term devices in noisy intermediate-scale quantum era. Here, we demonstrate that by implementing a multiple-output QNN that emulates a unital quantum channel, one can measure the expectation values of many incompatible observables simultaneously by Pauli-$Z$ measurements on distinct output qubits. We prove the existence of such quantum channel, derive analytical scaling constraints of the measured expectation values, and validate this framework by numerical simulations of observables learning tasks. Notably, our analysis reveals that it requires fewer copies of state when measuring some incompatible observables by the multiple-output QNNs, which demonstrates a resource efficiency advantage compared to separately applying projective measurements.
\end{abstract}

\maketitle

\section{Introduction}

Measurement plays a fundamental role in quantum mechanics. The Heisenberg uncertainty principle states that the variances of two observables $A$ and $B$ for any quantum state always satisfy the inequality $\Delta \hat{A} \Delta \hat{B} \geq (\langle[\hat{A},\hat{B}]\rangle/2i)^2$, which means two incompatible observables ($[\hat{A},\hat{B}]\neq 0$) cannot be simultaneously determined by a single type of measurement. 
The intrinsic uncertainty in quantum mechanics fundamentally arises from the collapse of a quantum state to one of the eigenstates of an operator under a projective measurement, also termed strong measurement. Two incompatible observables do not have shared eigenstates. Thus this collapse mechanism directly results in the erasure of all information about non-commuting observables during the measurement process.
Previous studies partially relaxed this restriction by introducing  sequential weak value measurements~\cite{PhysRevLett.95.220401,PhysRevLett.117.170402,PhysRevA.97.012122,Kim2018} or compressive sensing~\cite{PhysRevLett.112.253602} to measure incompatible observables. Weak values only extract a small amount of information from a single measurement and the quantum states basically do not collapse~\cite{RevModPhys.86.307}, which has been investigated theoretically~\cite{PhysRevLett.60.1351,PhysRevLett.92.130402,PhysRevA.76.062105,PhysRevLett.108.070402,PhysRevLett.66.1107,PhysRevLett.104.240401,PhysRevA.85.012107,PhysRevLett.113.200401} and experimentally using photons~\cite{PhysRevLett.94.220405,PhysRevLett.116.180401}.

Quantum machine learning (QML) and quantum neural networks (QNNs)~\cite{Cerezo2021,PhysRevLett.117.130501,PhysRevA.98.012324,Cong2019,Abbas2021,Cerezo2022,Biamonte2017} represent a novel intersection of quantum information and artificial intelligence, promising substantial improvements in quantum information processing capabilities.
The optimization of the parameterized quantum circuits in QNNs is to minimize the loss functions by strategies such as parameter shift rules~\cite{PhysRevLett.118.150503,PhysRevA.98.032309,PhysRevA.103.012405} and quantum natural gradient~\cite{Stokes2020quantumnatural,PhysRevA.106.062416}. Recent studies about precise expressivity of QNNs have enabled us to accurately represent operators using a QNN~\cite{PhysRevResearch.3.L032049,PhysRevLett.132.010602,Hou2023}. And the studies about learning physical properties of many observables~\cite{Huang2020,PhysRevLett.133.040202,PhysRevLett.129.240501} give the potential to combine the QNNs and quantum learning task.

Quantum information science has driven the development of QML and QNNs, which in turn will further advance the field of quantum information science. Here we propose a novel scheme to measure the expectation values of many incompatible observables using QNNs. We prove the existence of such multiple outputs QNNs by constructing a unital quantum channel. We also analytically compute the scaling restriction of the expectation values. The analytical derivation shows that it reduces the number of state copies for some observables. 

\section{Existence of the Unital Channel}

Let us introduce our model as follows. We construct a parameterized unital quantum channel
$\Phi_{\bm{\theta}}$ with trainable parameters $\bm{\theta}$, which is implemented as a QNN, applying on an $n$-qubit quantum state $\rho$. The unitality of $\Phi_{\mathbf{\theta}}$ implies that $\Phi_{\mathbf{\theta}}(I)=I$ with $I$ being the identity matrix.
The Pauli-$Z$ measurement is taken on each qubit of the output state 
$\Phi_{\bm{\theta}}(\rho)$. 
The QNN $\Phi_{\bm{\theta}}$ is learned from the datasets $\{  \rho_l, \tr(\rho_l O_i) \}_{l=1}^{L}$, $i = 1,2,\cdots, n_O$, such that for any $n_O \in [2,n]$ traceless Hermitian observables, each with eigenvalues whose absolute values do not exceed $1$, the expectation value of the $i$-th observable $\Tr (\rho O_i)$ equals to the expectation value of the Pauli-$Z$ measurement on the $i$-th qubit 
$\Tr (\Phi_{\bm{\theta}}(\rho) Z_i)$, 
up to a positive number
$\alpha\in \intervaloc{0,1}$, i.e., 
\begin{align}\label{eq1}
    \Tr \big(\Phi_{\bm{\theta}}(\rho) Z_i\big) = \alpha\Tr(\rho O_i), ~i = 1,\dots, n_O,
\end{align}
where $Z_i$ is the operator composed by the Pauli-$Z$ operator on the $i$-th qubit and the identity operators on all other qubits. 
The schematic of the multiple output QNN is shown in Fig.~\ref{fig_0}.
\begin{figure}[htbp]
  \vspace{-2mm}
  \includegraphics[width=6.14cm,height=3.5cm]{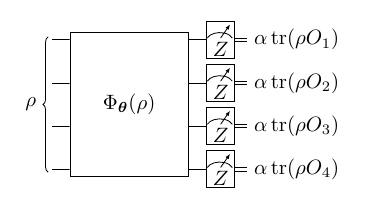}
  \vspace{-3mm}
  \caption{Schematic of a 4-qubit multiple output QNN for measuring incompatible
    observables.}
  \label{fig_0}
\end{figure}

The following theorem guarantees the existence of the multiple output QNN.

\begin{theorem} \label{thm_unital}
For any two traceless Hermitian operators $O_1$ and $O_2$, where each has eigenvalues whose absolute values bounded by $1$, there exists a unital quantum channel $\Phi$ and a number
$\alpha\in \intervaloc{0,1}$,
such that $\alpha O_1 = \Phi^{\dagger}(Z_1)$ and $\alpha O_2 = \Phi^{\dagger}(Z_2)$, where $Z_j$ is the Pauli-$Z$ operator on the $j$-th qubit.
\end{theorem}

\begin{proof}--
  To prove the existence of a quantum channel $\Phi$, we may show the existence of its adjoint channel $\Phi^{\dagger}$. Note that the trace preservation of channel $\Phi$ implies that $\Phi^{\dagger}(I)=I$. 
  In addition the channel $\Phi^{\dagger}$ must satisfy two equations $\Phi^{\dagger}(Z_1) = \alpha O_1$ and $\Phi^{\dagger}(Z_2) = \alpha O_2$. 
  To specify the adjoint channel $\Phi^{\dagger}$, we need to specify how $\Phi^{\dagger}$ acting on all the Pauli group elements $\{M_j\}$. Here we assume $\Phi_{\alpha}^{\dagger}(M_j)=0$, where $M_j$ is any Pauli group element expect $M_0=I$, $M_1=Z_1$, and $M_2=Z_2$.

  The Choi representation of the adjoint channel $\Phi^{\dagger}$ is
  \begin{align}
    \label{eq:1}
   J(\Phi_{\alpha}^{\dagger}) & = 4^n\Phi_{\alpha}^{\dagger}\otimes \mathbb{I} (|\phi\rangle\langle\phi|) \nonumber\\
   & = \sum_j \Phi_{\alpha}^{\dagger}(M_j) \otimes M_j \nonumber\\
   & =  I \otimes I + \alpha O_1 \otimes Z_1 + \alpha O_2 \otimes Z_2,
  \end{align}
  where the $n$-qubit maximally entangled state $|\phi\rangle = \sum_{j=0}^{2^n-1} |jj\rangle /\sqrt{2^n}$ is defined on the product Hilbert space $\mathcal{H}_{\mathcal{Y}}\otimes \mathcal{H}_{\mathcal{X}}$, and its density matrix is $|\phi\rangle\langle\phi|=\sum_j M_j\otimes M_j/4^n$. 
  We use Choi representation of $\Phi^{\dagger}$ to prove the completely positive of $\Phi^{\dagger}$ by Theorem 2.22 in Ref.~\cite{Watrous_2018}: $\Phi^{\dagger}$ is completely positive iff $J(\Phi^{\dagger})\ge0$.

  To prove $J(\Phi^{\dagger})\ge0$, let us denote the smallest eigenvalue of operator $A$ with $\lambda_{\text{min}}(A)$. Then
  \begin{align}
    \label{eq:3}
   & \lambda_{\text{min}}(J(\Phi_{\alpha}^{\dagger})) \nonumber\\
    & = 1 + \alpha \lambda_{\text{min}}(O_1 \otimes Z_1 +  O_2 \otimes Z_2) \nonumber\\
    & \ge 1 + \alpha \left( \lambda_{\text{min}}(O_1 \otimes Z_1) + \lambda_{\text{min}}(O_2 \otimes Z_2) \right) \nonumber\\
    & \ge 1 - 2 \alpha,
  \end{align}
  where we have used the Weyl inequality~\cite{MatrixAnalysis} in the third line.
  Thus when $0< \alpha \le \frac{1}{2}$, $\lambda_{\text{min}}(J(\Phi_{\alpha}^{\dagger})) \ge 0$, and $ J(\Phi_{\alpha}^{\dagger}) \ge 0$. This shows that $\Phi_{\alpha}^{\dagger}$ is completely positive. Then $\Phi_{\alpha}^{\dagger}$ has a Kraus representation, which implies that $\Phi_{\alpha}$ has an adjoint Kraus representation, and it is also completely positive.

  Following Eq.~\eqref{eq:1}, we obtain $\Tr_{\mathcal{Y}} J(\Phi_{\alpha}^{\dagger}) = I$. Hence $\Phi_{\alpha}$ ($0<\alpha\le \frac{1}{2}$) is unital by Theorem 2.26 in Ref.~\cite{Watrous_2018}. Therefore we complete the proof.

\end{proof}

From the above proof, Theorem~\ref{thm_unital} can be generalized to the cases with more than two traceless Hermitian operators.

\section{Determine Maximal $\alpha$}

From Eq.~\eqref{eq:1}, we obtain that the measurement fluctuations of $Z_1$ and $Z_2$ will decrease with the increasing of $\alpha$. Hence it is reasonable to optimize the QNN to find the maximal $\alpha$, denoted by $\alpha_{\max}$. From the above proof, for any two traceless operators $O_1$ and $O_2$ with $\norm{O_1}\le 1$ and $\norm{O_2}\le 1$, we obtain $\alpha_{\text{max}}\ge\frac{1}{2}$.

For a general unital channel $\Phi_{\alpha\boldsymbol{\beta}}$, we need to specify that for $j\notin\{0,1,2\}$,
\begin{equation}
  \label{eq:4}
 \Phi_{\alpha\boldsymbol{\beta}}^{\dagger}(M_j) = \sum_{k\neq 0} \beta_{k j} M_k.
\end{equation}
Then the Choi representation of $\Phi^{\dagger}$ is
\begin{align}
  \label{eq:5}
  J(\Phi_{\alpha\boldsymbol{\beta}}^{\dagger}) = 
  I \otimes I + \alpha \sum_{i=1}^2 O_i \otimes Z_i +  \sum_{j\notin\{0,1,2\}}^{k\neq 0} \beta_{k j} M_k \otimes M_j.
\end{align}
The task is to find an optimal $\Phi^{\dagger}$ satisfying
\begin{equation}
  \label{eq:6}
  \alpha_{\text{max}} = \{\max_{\boldsymbol{\beta}} \alpha: J(\Phi_{\alpha\boldsymbol{\beta}}^{\dagger}) \ge 0\},
\end{equation}
which can be expressed as a problem in semidefinite programming~\cite{gartner2012approximation}, see details in Appendix~\ref{app_maximal_alpha_method}.

To find $\alpha_{\text{max}}$, we design the following iterative algorithm. First, let us define the Choi representation of $\Phi^{\dagger}$ for the $m$-th iteration,
\begin{align}
      \label{eq:7}
    &J(\Phi_{\alpha\boldsymbol{\beta}}^{(m)\dagger}) \nonumber\\
    &=  I \otimes I 
    + \alpha^{(m)} \sum_{i=1}^2 O_i \otimes Z_i
    +  \sum_{j\notin\{0,1,2\}}^{k\neq 0} \beta^{(m)}_{k j} M_k \otimes M_j.
\end{align}

Our strategy is to choose suitable $\beta_{kj}^{(m)}$ such that 
$\alpha^{(m)}\ge \alpha^{(m-1)}$. 
Initially, $m=0$, we take $\beta^{(0)}_{kj}=0$, and $J(\Phi^{(0)\dagger}_{\alpha\boldsymbol{\beta}})=J(\Phi^{\dagger}_{\alpha})$. In the $m$-th step, we use $\lambda_{\text{min}}(J(\Phi^{(m)\dagger}_{\alpha\boldsymbol{\beta}}))=0$ to determine 
$\alpha^{(m)}$, 
and then solve the ground state of $J(\Phi^{(m)\dagger}_{\alpha\boldsymbol{\beta}})$: $\rho^{(m)}_G=\frac{1}{n^{(m)}_d}\sum_{i=1}^{n^{(m)}_d}|g^{(m)}_i\rangle\langle g^{(m)}_i|$, where $n^{(m)}_d$ is the ground state degeneracy. Next calculate $M_{kji}\equiv\langle g^{(m)}_i|M_k\otimes M_j|g^{(m)}_i\rangle$. If for any $i$, $M^{(m)}_{kji}$ have the same sign $s(mkj)\in\{-1,0,+1\}$, then we take the increment $\Delta^{(m)}_{kj}s(mjk)\ge0$, and $\beta^{(m+1)}_{kj}=\beta^{(m)}_{kj}+\Delta^{(m)}_{kj}$. Otherwise, $\Delta^{(m)}_{kj}=0$. When the step $m$ becomes larger, 
$\alpha^{(m)}$
limits to $\alpha_{\text{max}}$.

We can prove by perturbation theory that this method can always find a maximum $\alpha_{\text{max}}$. The proof and implement details of this method are presented in Appendix \ref{app_maximal_alpha_method}. The numerical results in the following section show that the $\alpha_{\text{max}}$ computed by our method is the same as the $\alpha$ obtained by optimizing the loss function in QNNs.

\section{Analysis of Samples Complexity}

Suppose we have obtained a well-trained optimized unital channel with $\alpha_{\text{max}}$, we will use it as a QNN to measure the expectation values of two 2-qubit observables $O_1$ and $O_2$ on a state $\rho$.
For direct projective measurements without QNNs, each measurement returns an eigenvalue, denoted as $\hat{o}_{1}$ for $O_1$ and $\hat{o}_{2}$ for $O_2$. In contrast, when applying a QNN, the Pauli-$Z$ measurements on the output qubits return outcomes $\hat{z}_{1},\hat{z}_{2} \in \{+1,-1\}$.
Define $N_O = N_{O_1} + N_{O_2}$ as the total number of copies of state $\rho$ prepared for projective measurement of $O_{1}$ and $O_{2}$, where $N_{O_j}$ is the sample size allocated to each observable. Let $N_Z$ denotes the number of copies of $\rho$ used for Pauli-$Z$ measurements with QNN. The estimate of the $j$-th observable expectation value is $\hat{O}_{j} = \sum_{i=1}^{N_{O_{j}}} \hat{o}_{j,i} /{N_{O_{j}}}$ for projective measurements, and $\hat{Z}_{j} = \sum_{i=1}^{N_{Z}} \hat{z}_{j,i} /(\alpha_{\text{max}}{N_{Z}})$ for Pauli $Z$ measurements with QNNs, where $\hat{o}_{j,i}$ ($\hat{z}_{j,i}$) is an eigenvalue of $O_j$ ($Z_j$) returned in the $i$-th projective measurement of $O_j$ ($Z_j$). When measurement numbers $N_O$ and $N_Z$ approach infinity, the limits of $\hat{Z}_j$ and $\hat{O}_j$ are denoted as $\mathbb{E}[\hat{z}_{j}/\alpha_{\text{max}}]$ and $\mathbb{E}[\hat{o}_{j}]$ respectively for $j \in \{1,2\}$. Following Eq.~\eqref{eq1}, $\mathbb{E}[\hat{z}_{j}/\alpha_{\text{max}}]=\mathbb{E}[\hat{o}_{j}] = \tr(\rho O_j)$. The variance of a random variable $\hat{o}$ is $\text{Var} [\hat{o}] = \mathbb{E}[\hat{o}^2] - \mathbb{E}[\hat{o}]^2$. We derive two conclusions in the following theorem.

\begin{theorem}\label{thm_complexity}
    (i) If we only focus on measuring one observable, the variance of Pauli-$Z$ measurements is always greater than direct projective measurements, 
    \begin{align}
        \text{Var} [\hat{z}_j/\alpha_{\max}] \geq \text{Var} [\hat{o}_j],~j=1,2.
    \end{align}
    (ii) Suppose by preparing $N_O$ copies of 
    $\rho$
    and taking projective measurements, with high probability, it achieves
    $\big|\hat{O}_{j}-\mathbb{E}[\hat{o}_{j}]\big|\leq \epsilon, ~j = 1,2.$
    Then, by preparing $N_Z = \mathcal{O}\big(\lambda N_O\big)$ copies of $\rho$, where
    \begin{align}\label{N_Z_without_Haar}
        \lambda = \frac{ 1  - \alpha_{\max}^2 \min\{\mathbb{E}[\hat{o}_1]^2,\mathbb{E}[\hat{o}_2]^2\}}{\alpha_{\max}^2[\text{Var}[\hat{o}_1] + \text{Var}[\hat{o}_2]]},
    \end{align}
    with high probability, the estimator $\hat{Z}_{j}$ achieves 
    $\big|\hat{Z}_{j}-\mathbb{E}[\hat{o}_j]\big|\leq \epsilon, ~j = 1,2$.
    
   If $\rho = |\psi \rangle\langle\psi|$ is a pure state, the average number of copies under Haar measure becomes $N_Z = \mathcal{O}\big(\lambda_{H}N_O\big)$,
   where 
   \begin{align}\label{N_Z_with_Haar}
       \lambda_{H} = \frac{d(d+1)-\alpha_{\max}^2 \min \{\tr(O_1^2),\tr(O_2^2)\}}{d\alpha_{\max}^2[\tr(O_1^2)+\tr(O_2^2)]},
   \end{align}
   and $d$ is the dimension of the Hilbert space.
\end{theorem}

The proof details are presented in Appendix~\ref{app_proof_number}. The result $(i)$ states that if we only measure one observable or focus on the output of a single qubit, to achieve the same measurement accuracy, the QNN method cannot decrease the number of state copies compared with direct projective measurement. On the other hand, result $(ii)$ shows that, the $\lambda$ in Eq.~\eqref{N_Z_without_Haar}, or the $\lambda_H$ in Eq.~\eqref{N_Z_with_Haar}, determines whether the QNN approach can reduce the number of state copies when measuring two observables.
Following this analysis, the advantages of QNN method becomes apparent only when simultaneously measuring multiple observables. In the following numerical results, we provide an example of two observables for which the QNN decreases the sample complexity.

\section{Architecture of the mixed-unitary channel}

A unital channel is termed a mixed-unitary channel if it can be expressed as $\mathcal{E}_{\bm{\theta},\vec{\omega}}(\cdot) = \sum_{i} w_{i} U_{i} \cdot U_{i}^{\dagger}$, where $U_{i}$ has trainable parameters $\bm{\theta}_i$, and $w_i$ is the probability to perform the unitary transformation $U_i$ satisfying $\sum_i w_{i} = 1$~\cite{Watrous_2018}. In the model training process, we employ a parameterized mixed-unitary channel as the ansatz to approximate the target unital channel $\Phi_{\bm{\theta}}$. An architecture of a mixed-unitary channel for a 2-qubit input state $\rho$ is shown in Fig.~\ref{fig_circuit}.
\begin{figure}[htbp]
\vspace{-3mm}
\includegraphics[width=8.0cm,height=2.8cm]{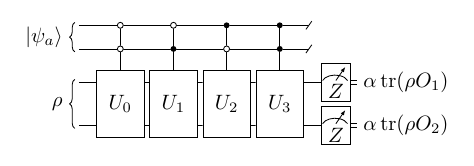}
\vspace{-3mm}
\caption{An illustration of mixed-unitary channel acting as multiple output QNNs for a $2$-qubit state $\rho$. 
The $2$-qubit ancillary state $|\psi_a\rangle = \sqrt{w_{0}}|00\rangle + \sqrt{w_{1}}|01\rangle + \sqrt{w_{2}}|10\rangle + \sqrt{w_{3}}|11\rangle$. 
The open circle notation means a control gate conditioning on the qubit being set to $|0\rangle$, while a closed circle indicates conditioning on the qubit being set to $|1\rangle$.}
\label{fig_circuit}
\end{figure}

The circuit requires an ancillary state $|\psi_a\rangle = \sum_{i=0}^{d_a-1}\sqrt{w_{i}}|i\rangle$,
where $d_a$ is the Hilbert space dimension of the ancillary system, $|i\rangle$ is the $i$-th computational basis, and the set of weights $\{w_i\}$ are trainable parameters with $\sum_i w_i = 1$. Then a control-$U$ gate is implemented, $CU = \sum_{i=0}^{d_a-1} |i\rangle\langle i| \otimes U_i$,
where $U_i$ is a unitary gate of the $n$-qubit system conditioning on the $i$-th computational basis $|i\rangle$ of the ancillary state.
Before measurement, the combined state becomes $\rho_{T} = \sum_{i,j=0}^{d_a-1} \sqrt{w_{i}w_{j}} |i\rangle\langle j| \otimes U_{i}\rho U_{j}^{\dagger}$.
After trace out the ancillary state, the state to be measured is $\Tr_a \rho_T = \sum_{i=0}^{d_a-1} w_{i} U_{i} \rho U_{i}^{\dagger} = \mathcal{E}(\rho)$.
At the end of the circuit we make the Pauli-$Z$ measurements on every qubit of $\mathcal{E}(\rho)$. By optimizing the parameters $\{w_{i}\}$ and $\{\bm{\theta}_i\}$, we maximize the parameter $\alpha$ requiring that the expectation value of $\Tr(\rho O_j)$  equals to the expectation value of $Z_j$ on $\mathcal{E}(\rho)$ up to the positive number $\alpha$, i.e., $\alpha \Tr(\rho O_j) = \Tr (\mathcal{E}(\rho) Z_j)$ with $Z_1=Z\otimes I$ and $Z_2=I\otimes Z$. Note that the above equations can be rephrased in the following equivalent form
\begin{align}
    \alpha O_j  = \sum_{i=0}^{d_a-1} w_{i} U_{i}^{\dagger} Z_j U_{i} = \mathcal{E}^{\dagger}(Z_j).\label{O1}
\end{align}
We point out that such quantum channel can also be realized by a randomness-enhanced QNN~\cite{PhysRevLett.132.010602}.

\section{Numerical Results}

Suppose we have two data sets $\{ ( \rho_l^{(1)}, \mathcal{O}_l^{(1)} ) \}_{l=1}^{L}$ and $\{ ( \rho_m^{(2)}, \mathcal{O}_m^{(2)} ) \}_{m=1}^{M}$. We define the loss function as
\begin{align}\label{loss_function}
    \mathcal{L} (\bm{\theta},\vec{w},\alpha) &= \mathcal{L}_O(\bm{\theta},\vec{w}) +  \mathcal{L}_{\alpha}(\alpha) \nonumber \\
    &= \frac{1}{L} \sum_{l=1}^L \bigg\{\alpha\mathcal{O}_l^{(1)} - \tr(Z_1 \mathcal{E}_{\bm{\theta},\vec{w}}(\rho_l^{(1)})) \bigg\}^2 \nonumber \\ 
    &+ \frac{1}{M} \sum_{m=1}^M \bigg\{\alpha\mathcal{O}_m^{(2)} - \tr(Z_2 \mathcal{E}_{\bm{\theta},\vec{w}}(\rho_m^{(2)})) \bigg\}^2 - \alpha.
\end{align}
The term $\mathcal{L}_{\alpha} = -\alpha$ is used to maximize $\alpha$.
The gradient descending optimization process of $\mathcal{L}$ by Adam optimizer is shown in Fig.~\ref{fig_operator_learning}(a). The loss function $\mathcal{L}_{O}$ as a function of training epoch for different $d_a$ is plotted. The optimized circuit approximates the unital channel well when $d_a \geq 3$.

We note that mixed-unitary channel is related to the Uhlmann theorem~\cite{nielsen2002introduction}, which states that there exists a mixed-unitary channel $\mathcal{E}$ such that $\mathcal{E}(A) = B$ if and only if $A\succ B$. And $A\succ B$ if and only if $\lambda_{A} \succ \lambda_{B}$, where $\lambda_{O}$ is the vector of eigenvalues for the operator $O$ in descending order. Here the majorization $\vec{a} \prec \vec{b}$ for two $d$-dimensional vectors $\vec{a}$ and $\vec{b}$ whose components arranged in descending order is defined as (i) $\sum_{i=1}^{d'} a_i \leq \sum_{i=1}^{d'} b_i $, $1\leq d' \leq d$, and (ii) $\sum_{i=1}^d a_i = \sum_{i=1}^d b_i = \text{constant}$. In the case of two operators $\alpha O_1 = \mathcal{E}(Z_1)$ and $\alpha O_2 = \mathcal{E}(Z_2)$, we identify three majorization constraints, which are $Z_1 \succ \alpha O_1 $, $Z_2 \succ \alpha O_2 $ and  $xZ_1 + yZ_2 \succ \alpha( xO_1 + yO_2) $, with $x$ and $y$ being any real numbers. The details are analyzed in Appendix \ref{app_majorization}. In the previous randomness enhanced QNN \cite{PhysRevLett.132.010602}, the Uhlmann theorem can be used to prove the existence of the mixed unitary channel which only learns one observable $\mathcal{E}^{\dagger}(Z_1) = O$. However in our model with learning two observables, this majorization constraints is not tight compared with the completely positive constraints of the unital channel. In Fig.~\ref{fig_operator_learning}(b), we plot the maximal $\alpha$ obtained from numerical optimizations, analytical computed $\alpha_{\text{max}}$, as well as the $\alpha_{\text{maj}}$ derived only from the majorization constraints. The maximal $\alpha$ from numerical optimization of QNN matches theoretical predictions, while the majorization constraints are not tight, i.e., $\alpha = \alpha_{\text{max}} \leq \alpha_{\text{maj}}$.

\begin{figure}[htbp]

\subfigure[]{
\includegraphics[width=4.5cm,height=4.2cm]{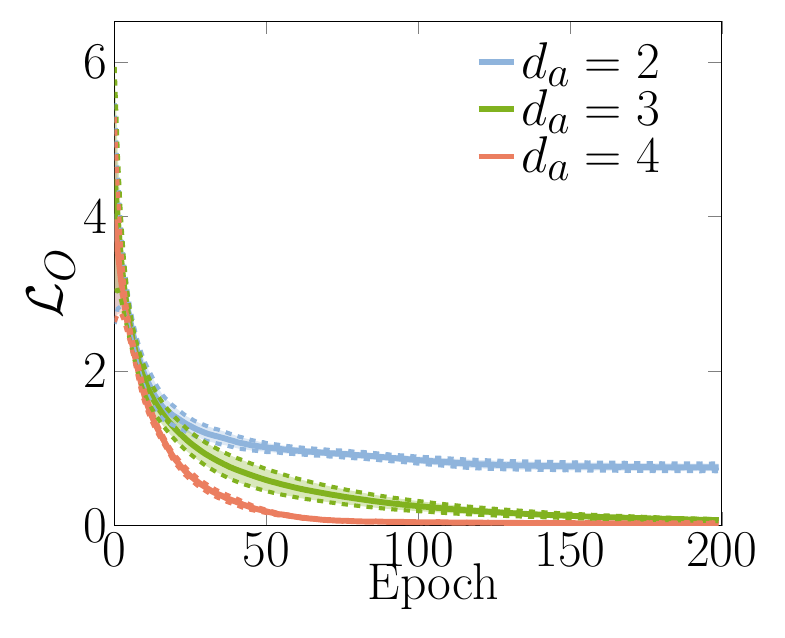}
}
\hspace{-8mm}
\subfigure[]{
\includegraphics[width=4.3cm,height=4.13cm]{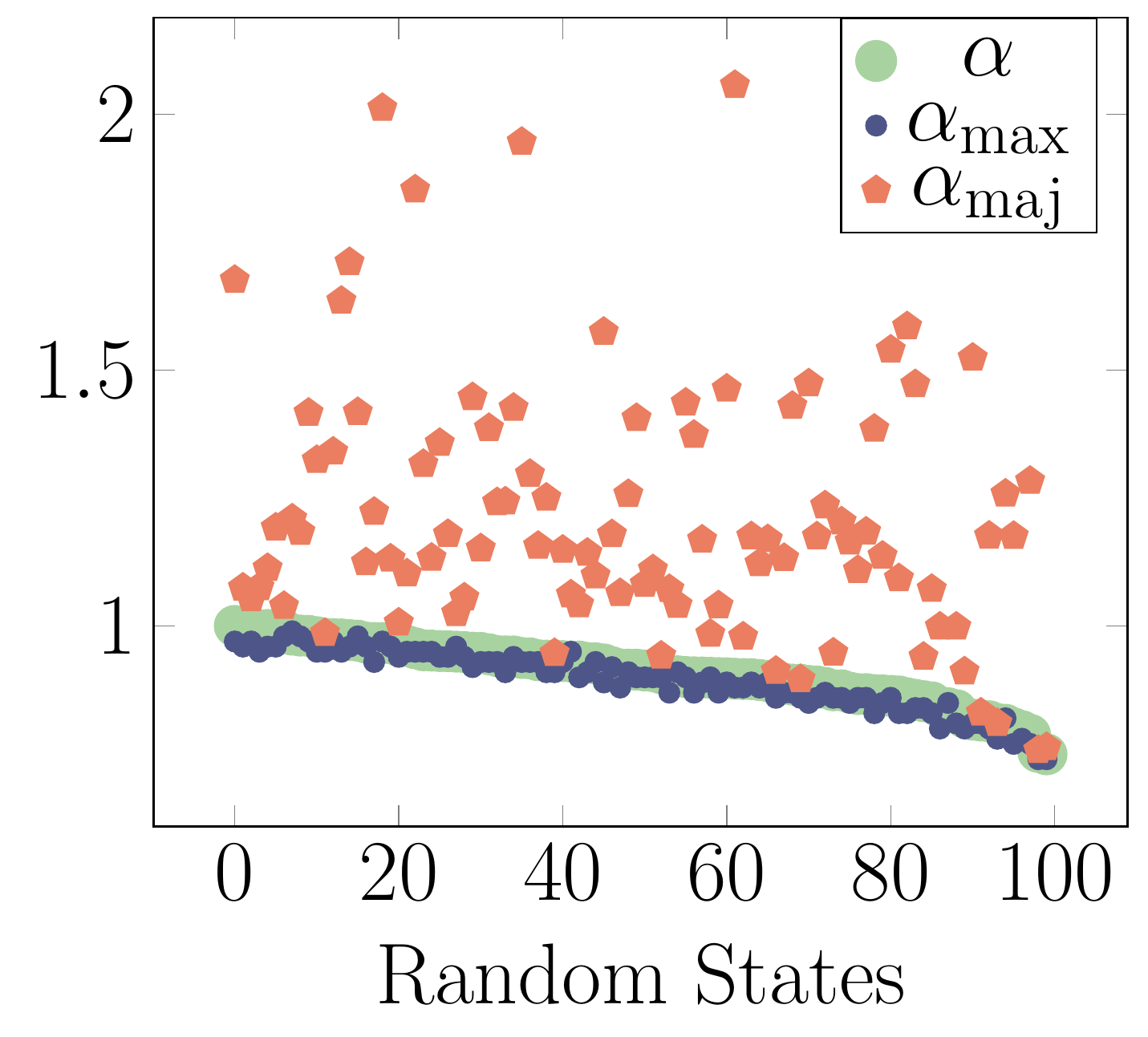}
}
\vspace{-3mm}
\caption{
(a) Learning 2-qubit operators by optimizing $\mathcal{L}$ for $d_a = 2,3$ and $4$. The size of data set is $L = M = 100$, and the learning rate is $0.05$. 
The $\mathcal{L}_O$ is averaged by $100$ pairs of random operators $O_1$ and $O_2$. 
(b) The comparison between numerical results $\alpha$, analytical methods $\alpha_{\text{max}}$, and majorization limitations $\alpha_{\text{maj}}$ for $100$ random operators with $d_a = 4$. 
The order of $\alpha$ is rearranged as descending order. 
}

\label{fig_operator_learning}
\end{figure}

\begin{figure}[htbp]
\vspace{-3mm}
\includegraphics[width=6.0cm,height=6.0cm]{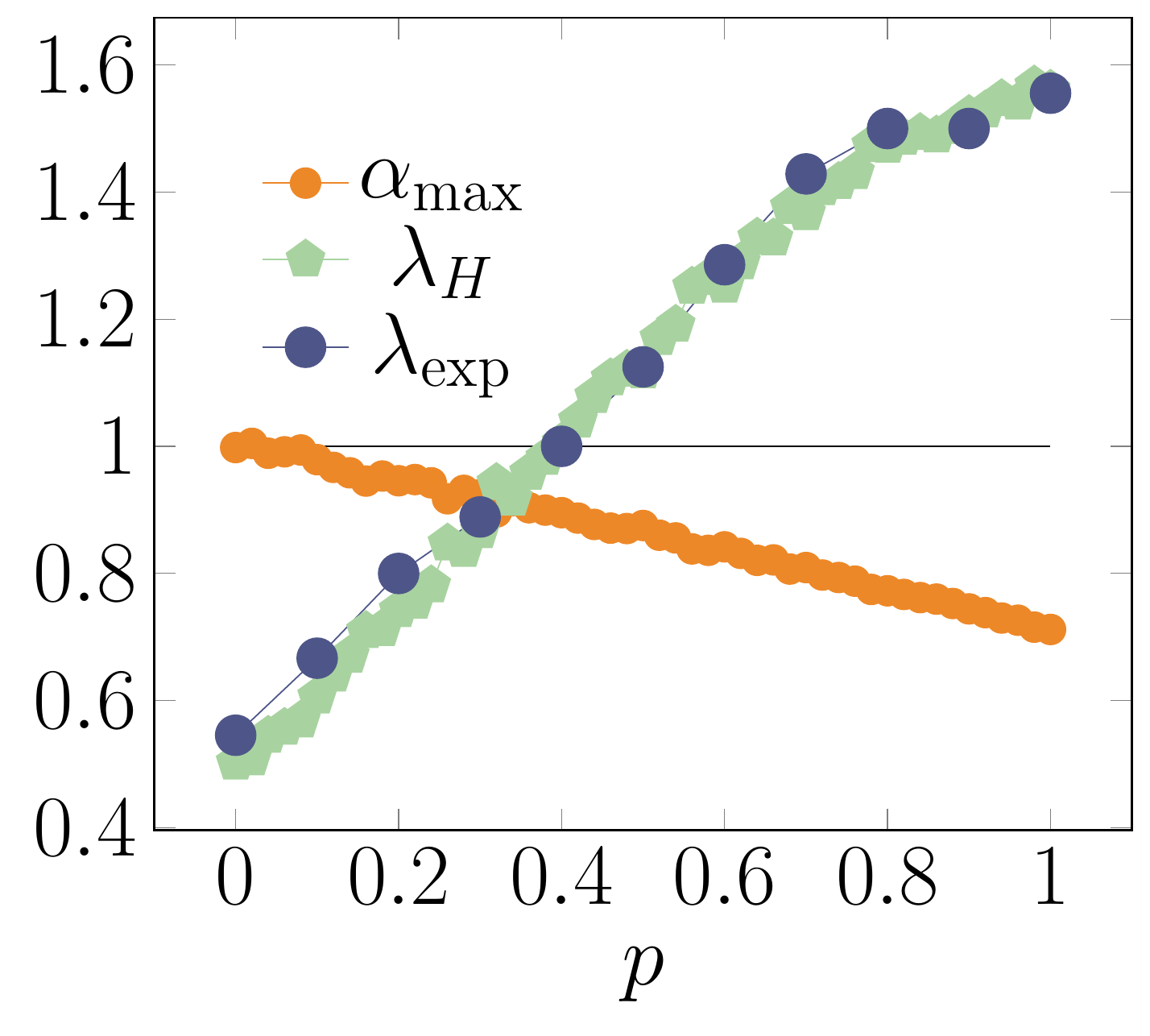}
\vspace{-3mm}
\caption{
The ratio $\lambda_{H}$, the number $\alpha_{\max}$ and the simulation results $\lambda_{\text{exp}}$ as the function of $p$.
The $\alpha_{\max}$ is obtained by minimizing the loss function of Eq.~\eqref{loss_function}. 
Each point of $\lambda_{\text{exp}}$ takes average of $1000$ random pure states. 
}
\label{fig_copies_number}
\end{figure}

We also numerically investigate the sample complexity of learning multiple observables. We construct an example of a series of observables $O_1(p) = (1-p) Z_1 Z_2 + (p/2) (Z_1 + Z_2)$ with $p\in[0,1]$, and $O_2 = X_1X_2$. When $p = 0$, the operators commute $[O_1(0),O_2] = 0$; when $p > 0$, operators do not commute $[O_1(p),O_2] \neq 0$. There exists a $p_c$ such that, when $p>p_c$, using QNNs requires more copies than classical projective measurements. Whereas for $p<p_c$, the QNNs outperform classical methods. In Fig.~\ref{fig_copies_number}, we plot the $\lambda_{H}$ as the function of $p\in [0,1]$. It shows that the QNNs consume fewer copies for $p<0.4$.
We numerically simulate the measurements with QNNs and respective projective measurements on $O_1$ and $O_2$, and plot the $\lambda_{\text{exp}} = N_{Z}^{\text{exp}}/N_O^{\text{exp}}$ in Fig.~\ref{fig_copies_number}, where the $N_{Z}^{\text{exp}}$ and $N_O^{\text{exp}}$ are the number of states copies required to get the expectation value with error less than $0.01$. When $p<0.4$, using QNN method consumes fewer copies than direct projective measurements.

Another example is taking $O_1 = O_2 = O$, where $O$ can be any observable. In this case, however,  the number of copies generally increases. The numerical results show that the $\lambda_{H}$'s for $100$ random observables $O$ are generally greater than $2$. It means that if we copy one observable into two replicas, it fails to reduce the number of state copies. To reduce the number of copies, $\lambda_H$ in Eq.~\eqref{N_Z_with_Haar} needs to be less that $1$, which becomes
\begin{align}\label{alpha_inequality}
    \alpha_{\max}^{2}\tr(O^2)> \frac{(d+1)d}{(2d+1)}>2,
\end{align}
where the second inequality is obtained when $d\geq4$. By the majorization constraints analyzed in details in Appendix~\ref{app_majorization}, in general $\alpha_{\max}$ and $\tr(O^2)$ can not be large simultaneously, which makes the inequality of Eq.~\eqref{alpha_inequality} impossible.

\section{Discussion and outlook}

We propose a multiple-output QNN tailored for simultaneous measurement of incompatible observables using only single-qubit Pauli-$Z$ measurements on spatially separated qubits. This QNN-based measurement protocol enables concurrent extraction of information of multiple non-commuting observables. The expectation value of the $i$-th observable equal to the expectation value of the Pauli-$Z$ measurement on the $i$-th qubit up to a number $\alpha_{\text{max}}$.
In our QNN method, the collapse of the measured qubit remains localized, thereby it can bypass the restrictions imposed by the Heisenberg uncertainty principle.
It is worthy to note that in our protocol we only get the expectation values $\langle Z_i \rangle = \alpha\langle O_i \rangle$, the collapsed states $|0\rangle$ and $|1\rangle$ on the reduced density matrix of $i$-th qubit do not have a direct connection with the eigenstates of observable $O_i$.

As a QML model, our model can learn many non-commuting observables in one quantum circuit. And it can reduce the sample complexity if we measure the expectation values for some observables by using the optimized circuit.
Our protocol is also a generalization of the randomness-enhanced QNN~\cite{PhysRevLett.132.010602}. This previous work proved the existence of a mixed-unitary channel for learning one observable. Our new protocol gives a more precise and general analytical results about learning many observables. Our method provides a general framework to analyze the expressive abilities and restrictions of quantum channel, and also offers the capacity for the design and development of more complex QMLs and QNNs schemes.

We believe this protocol can be generalized to other kinds of QML architectures, such as deep quantum neural networks~\cite{Beer2020}, quantum recurrent neural networks~\cite{PhysRevA.103.052414}, quantum convolutional neural networks~\cite{Cong2019} and quantum autoencoders~\cite{PhysRevLett.124.130502,huang2024optimizedquantumautoencoder}. It can also be extended from expectation value measurements to other quantum resources measure~\cite{RevModPhys.81.865,RevModPhys.80.517,GUHNE20091,RevModPhys.89.041003,HU20181,RevModPhys.74.197}. Furthermore, as the important application of quantum mechanics, QML and QNNs have shown the abilities to surpass classical computation. We believe that QML and QNNs will have substantial potential to impact quantum mechanics and quantum information science in future researches.

\begin{acknowledgments}
This work is supported by National Key Research and Development Program of China (Grants No.2021YFA0718302 and No. 2021YFA1402104).
\end{acknowledgments}

\bibliography{bib.bib}

\onecolumngrid
\newpage

\appendix

\section{Method of determining maximal $\alpha$}\label{app_maximal_alpha_method}

In this section we give the details of the method to compute the maximal $\alpha$. We will first introduce the framework of this approach. Then we prove that it can always obtain the optimal $\alpha_{\text{max}}$ using this method. Next we give some examples and explain each step in details. We also introduce the semidefinite programming formulation of the problem.

Recall our proof of the Theorem 1, for any two traceless Hermitian operators $O_1$ and $O_2$ with maximum absolute value of eigenvalues not exceeding $1$, there exists a unital quantum channel $\Phi$ and a number $\alpha\in \intervaloc{0,1}$,
such that $\alpha O_1 = \Phi^{\dagger}(Z_1)$ and $\alpha O_2 = \Phi^{\dagger}(Z_2)$, where $Z_j$ is the Pauli-$Z$ operator on the $j$-th qubit. 
The Choi representation of the adjoint channel $\Phi^{\dagger}$ can be constructed as
\begin{align}
J(\Phi_{\alpha}^{\dagger}) & = 4^n \Phi_{\alpha}^{\dagger}\otimes \mathbb{I} (|\phi\rangle\langle\phi|) \nonumber\\
& = \sum_j \Phi_{\alpha}^{\dagger}(M_j) \otimes M_j \nonumber\\
& = I \otimes I + \alpha O_1 \otimes Z_1 + \alpha O_2 \otimes Z_2,
\end{align}
where the $n$-qubit maximally entangled state $|\phi\rangle = \sum_{j=0}^{2^n-1} |jj\rangle /\sqrt{2^n}$ is defined on the product Hilbert space $\mathcal{H}_{\mathcal{Y}}\otimes \mathcal{H}_{\mathcal{X}}$, and its density matrix is $|\phi\rangle\langle\phi|=\sum_j M_j\otimes M_j/4^n$. Here we only focus on the existence of such $\Phi^{\dagger}$, so we assume $\Phi_{\alpha}^{\dagger}(M_j)=0$, where $M_j$ is any Pauli group element expect $M_0=I$, $M_1=Z_1$, and $M_2=Z_2$.

Next we focus on finding the maximal $\alpha$. For a general unital channel $\Phi_{\alpha\boldsymbol{\beta}}$ satisfying $\Phi_{\alpha\boldsymbol{\beta}}(Z_1) = O_1$ and $\Phi_{\alpha\boldsymbol{\beta}}(Z_2) = O_2$, we need to specify that for $j\notin\{0,1,2\}$
\begin{equation}
 \Phi_{\alpha\boldsymbol{\beta}}^{\dagger}(M_j) = \sum_{k\neq 0} \beta_{k j} M_k. 
\end{equation}

Then the Choi representation of $\Phi^{\dagger}$
\begin{align}
  J(\Phi_{\alpha\boldsymbol{\beta}}^{\dagger}) =   I \otimes I + \alpha \sum_{i=1}^2 O_i \otimes Z_i +  \sum_{j\notin\{0,1,2\}}^{k\neq 0} \beta_{k j} M_k \otimes M_j.
\end{align}

The task is to find an optimal $\Phi^{\dagger}$ satisfying
\begin{equation}\label{a_max_definition_app}
  \alpha_{\text{max}} = \{\max_{\boldsymbol{\beta}} \alpha: J(\Phi_{\alpha\boldsymbol{\beta}}^{\dagger}) \ge 0\}.
\end{equation}

\subsection{Framework of the Method}

To find $\alpha_{\text{max}}$, we design the following iterative algorithm. First, let us define the Choi representation of the $m$-th iteration as 
\begin{align}
    J(\Phi_{\alpha\boldsymbol{\beta}}^{(m)\dagger}) =  I \otimes I 
    + \alpha^{(m)} \sum_{i=1}^2 O_i \otimes Z_i
    +  \sum_{j\notin\{0,1,2\}}^{k\neq 0} \beta^{(m)}_{k j} M_k \otimes M_j.
\end{align}
Our strategy is to choose suitable $\beta_{kj}^{(m)}$ such that 
$\alpha^{(m)}\ge \alpha^{(m-1)}$. 
Initially, $m=0$, we take $\beta^{(0)}_{kj}=0$, and 
\begin{align}
    J(\Phi^{(0)\dagger}_{\alpha\boldsymbol{\beta}})=J(\Phi^{\dagger}_{\alpha}). 
\end{align}
In the $m$-th step, we use 
\begin{align}
    \lambda_{\text{min}}(J(\Phi^{(m)\dagger}_{\alpha\boldsymbol{\beta}}))=0
\end{align}
to determine 
$\alpha^{(m)}$, 
and then solve the ground state of $J(\Phi^{(m)\dagger}_{\alpha\boldsymbol{\beta}})$: 
\begin{align}
    \rho^{(m)}_G=\frac{1}{n^{(m)}_d}\sum_{i=1}^{n^{(m)}_d}|g^{(m)}_i\rangle\langle g^{(m)}_i|,
\end{align}
where $n^{(m)}_d$ is the ground state degeneracy. 

Next calculate 
\begin{align}
    M^{m}_{kji}\equiv\langle g^{(m)}_i|M_k\otimes M_j|g^{(m)}_i\rangle. 
\end{align}
If for any $i$, $M^{(m)}_{kji}$ have the same sign $s(mkj)\in\{-1,0,+1\}$, then we take the increment $\Delta^{(m)}_{kj}s(mjk)\ge0$, and 
\begin{align}
    \beta^{(m+1)}_{kj}=\beta^{(m)}_{kj}+\Delta^{(m)}_{kj}. 
\end{align}
Otherwise, $\Delta^{(m)}_{kj}=0$. The Choi representation in the $(m+1)$-th iteration is
\begin{align}
    J(\Phi_{\alpha\boldsymbol{\beta}}^{(m+1)\dagger}) = J(\Phi_{\alpha\boldsymbol{\beta}}^{(m)\dagger}) + \sum_{j\notin\{0,1,2\}}^{k\neq 0} \Delta^{(m)}_{kj} M_k\otimes M_j.
\end{align}
When the step $m$ becomes larger, 
$\alpha^{(m)}$ 
limits to $\alpha_{\text{max}}$.

If the ground state is not degenerated, $n^{(m)}_d = 1$, there is only one ground state of $J(\Phi^{(m)\dagger}_{\alpha\boldsymbol{\beta}})$: 
\begin{align}
    \rho^{(m)}_G=|g^{(m)}\rangle\langle g^{(m)}|.
\end{align}
We calculate 
\begin{align}
    M^{(m)}_{kj}\equiv\langle g^{(m)}|M_k\otimes M_j|g^{(m)}\rangle,
\end{align}
and we only need to find all $M_k\otimes M_j$ with $M^{(m)}_{kj}\neq 0$. The increment is $\Delta^{(m)}_{kj}s(mjk)\ge0$. 

\subsection{Proof of the Method}
We now prove that this method will always find a maximum $\alpha_{\text{max}}$ by perturbation theory. We focus on the non-degenerate case, and it is natural to generalize it to the degenerate cases. The total procedure of the $(m+1)$-th iteration is composed of two steps. The first step is to add a perturbation $\sum_{j\notin\{0,1,2\}}^{k\neq 0} \Delta^{(m)}_{kj} M_k\otimes M_j$ on the Choi representation in $m$-th iteration $J(\Phi_{\alpha\boldsymbol{\beta}}^{(m)\dagger})$, which is  
\begin{align}
    J(\Phi_{\alpha_0\boldsymbol{\beta}}^{(m+1)\dagger}) = J(\Phi_{\alpha_0\boldsymbol{\beta}}^{(m)\dagger}) + \sum_{j\notin\{0,1,2\}}^{k\neq 0} \Delta^{(m)}_{kj} M_k\otimes M_j.
\end{align}
Denote the ground state and excitation states of the $J(\Phi_{\alpha_0\boldsymbol{\beta}}^{(m)\dagger})$ as $|g^{(m)}\rangle$ and $|e_i^{(m)}\rangle$. We have $\langle g|J(\Phi_{\alpha_0\boldsymbol{\beta}}^{(m)\dagger})|g\rangle = 0$ and $\langle e_i|J(\Phi_{\alpha_0\boldsymbol{\beta}}^{(m)\dagger})|e_i\rangle > 0$. The first order perturbation of the ground state energy is 
\begin{align}
    E_{g_0}^{(m+1)} &= 0 + \langle g|\sum_{j\notin\{0,1,2\}}^{k\neq 0} \Delta^{(m)}_{kj} M_k\otimes M_j|g\rangle.
\end{align}
And we can always choose some small enough $\Delta^{(m)}_{kj}$ to make other excitation energies still greater than $0$. 
And we can get a positive definite diagonalized Hamiltonian in the basis of $J(\Phi_{\alpha_0\boldsymbol{\beta}}^{(m+1)\dagger})$.

The second step of the $(m+1)$-th iteration is to add another perturbation $\delta\alpha (O_1 \otimes Z_1 +  O_2 \otimes Z_2)$ on the Hamiltonian $J(\Phi_{\alpha_0\boldsymbol{\beta}}^{(m+1)\dagger})$, and get $\delta\alpha$ by solving the equation 
\begin{align}
    \lambda_{\text{min}}\bigg[ J(\Phi_{\alpha_0\boldsymbol{\beta}}^{(m+1)\dagger}) + \delta\alpha (O_1 \otimes Z_1 +  O_2 \otimes Z_2) \bigg]=0.
\end{align}
We can simplify the function above into the following function. The ground state of $J(\Phi_{\alpha_0\boldsymbol{\beta}}^{(m+1)\dagger})$ is $|g_0^{(m+1)}\rangle$ and the ground state energy of the first order perturbation is
\begin{align}
    E_{g}^{(m+1)} &= E_{g_0}^{(m+1)} + \langle g_0^{(m+1)}|\delta\alpha (O_1 \otimes Z_1 +  O_2 \otimes Z_2)|g_0^{(m+1)}\rangle.
\end{align}
We can always find a $\delta\alpha>0$ that satisfies the function $E_{g}^{(m+1)} = 0$. Thus completes the proof.

\subsection{Some Examples of the Method}
Next we take three examples, the first is that $O_1=X_1X_2$ and $O_2=(Z_1+Z_2)/2$, the second is $O_2$ and $O_3=\frac{1}{3}(X_1X_2 - X_1Z_2 + I_1Y_2)$, and the third is two random observables. In the example of $O_1$ and $O_2$, the Choi representation in the 
\begin{align}
    J(\Phi_{\alpha\boldsymbol{\beta}}^{(0)\dagger}) =  I  + \frac{1}{2}\alpha ZIZI + \frac{1}{2}\alpha IZZI + \alpha XXIZ. 
\end{align}

The ground state of $J(\Phi_{\alpha\boldsymbol{\beta}}^{(0)\dagger})$ is $4$-fold degenerated, $\rho^{(0)}_G=\frac{1}{4}\sum_{i=1}^{4}|g^{(0)}_i\rangle\langle g^{(0)}_i|$. By calculating all $M^{(0)}_{kji}$'s, we found that there does not exist an $M_k\otimes M_j$ such that $M^{(0)}_{kji}$'s have the same sign $s(mkj)\in\{-1,0,+1\}$ for all $i$. Thus in the case of $O_1$ and $O_2$, we obtain the $\alpha_{max}$ by solving the smallest eigenvalue of $J(\Phi_{\alpha\boldsymbol{\beta}}^{(0)\dagger})$ is equal to $0$, i.e., $\lambda_{\text{min}}(J(\Phi_{\alpha\boldsymbol{\beta}}^{(0)\dagger})) = 0$. And we get $\alpha_{\text{max}} = \frac{\sqrt{2}}{2}$.

In the example of $O_2$ and $O_3$, the Choi representation is
\begin{align}
    J(\Phi_{\alpha\boldsymbol{\beta}}^{(0)\dagger}) =  I  + \frac{1}{2}\alpha IZZI + \frac{1}{2}\alpha ZIZI + \frac{1}{3}\alpha IYIZ + \frac{1}{3}\alpha XXIZ -\frac{1}{3}\alpha XZIZ. 
\end{align}

The ground state of the $J(\Phi_{\alpha\boldsymbol{\beta}}^{(0)\dagger})$ is also $4$-fold degenerated. By calculating all $M^{(0)}_{kji}$'s, we found that there are two terms, $XI\otimes ZZ$ and $ZY\otimes ZZ$, satisfying that $M^{(0)}_{kji}$'s have the same sign. $s(0,XI,ZZ) = -1$ and $s(0,ZY,ZZ) = 1$ for $i=1,2,3,4$. In the real implementation, we found that the terms $M^{(m)}_{kji}$ with the same sign are invariant for all $m$-th iterations, i.e., there are only $XI\otimes ZZ$ and $ZY\otimes ZZ$ that have the same sign of the terms $M^{(m)}_{kji}$ for all $m$. So we only need to find the first $M^{(0)}_{kji}$'s with the same sign $s(0kj)$ the and set the $\beta_{kj}$'s as variational parameters and other $\beta_{k'j'}$'s are set to be $0$.

Thus the Choi representation in the $m$-th iteration is 
\begin{align}
    J(\Phi_{\alpha\boldsymbol{\beta}}^{(m)\dagger}) &= J(\Phi_{\alpha\boldsymbol{\beta}}^{(0)\dagger}) + \beta_{XI,ZZ} XI\otimes ZZ + \beta_{ZY,ZZ} ZY\otimes ZZ \nonumber \\
     &=   I  + \frac{1}{2}\alpha IZZI + \frac{1}{2}\alpha ZIZI + \frac{1}{3}\alpha IYIZ + \frac{1}{3}\alpha XXIZ -\frac{1}{3}\alpha XZIZ \nonumber\\
     &+ \beta_{XI,ZZ} XI\otimes ZZ + \beta_{ZY,ZZ} ZY\otimes ZZ.
\end{align}
By the above iteration method, the $\alpha_{\text{max}} \approx 0.927$, and $\beta_{XI,ZZ} \approx -0.062 $, $\beta_{ZY,ZZ} \approx 0.062 $. If we solve the function $\lambda_{\text{min}}[J(\Phi_{\alpha\boldsymbol{\beta}}^{(0)\dagger})] = 0$, we can get $\alpha_0 = \sqrt{6/(5+\sqrt{17})} \approx 0.811$, which is less than the $\alpha_{\text{max}}$.

In the third example, these two observables are two random observables $O^{\text{random}}_1$ and $O^{\text{random}}_2$. In this case the ground state of the Choi representation is usually not degenerated. So there is only one ground state $\rho^{(0)}_G = |g^{(0)}\rangle\langle g^{(0)}|$. Usually all operators $M_k\otimes M_j$ for all $k,j$ satisfy $M^{(0)}_{kj}=\langle g^{(0)}|M_k\otimes M_j|g^{(0)}\rangle \neq 0$. So all $\beta_{kj}$'s are set as variational parameters, and the Choi representation in the $m$-th iteration is 
\begin{align}
    J(\Phi_{\alpha\boldsymbol{\beta}}^{(m)\dagger}) = J(\Phi_{\alpha\boldsymbol{\beta}}^{(0)\dagger}) + \sum_{j\notin\{0,1,2\}}^{k\neq 0} \beta^{(m)}_{kj} M_k\otimes M_j.
\end{align}
In real implementation of the third example, we found that $M^{(0)}_{kji} = \langle g^{(m)}_i|M_k\otimes M_j|g^{(m)}_i\rangle \neq 0$ is possible only for $M_j = ZZ$. So the above equation can be simplified to 
\begin{align}
    J(\Phi_{\alpha\boldsymbol{\beta}}^{(m)\dagger}) = J(\Phi_{\alpha\boldsymbol{\beta}}^{(0)\dagger}) + \sum_{j\notin\{0,1,2\}} \beta^{(m)}_{k,ZZ} M_k\otimes ZZ.
\end{align}

In real implementation, we can also use another optimization method to get the $\alpha_{\text{max}}$. We first get the $\alpha_0$ by solving the function $\lambda_{\text{min}}[J(\Phi_{\alpha\boldsymbol{\beta}}^{(0)\dagger})] = 0$. And then we add the variational parameters $\beta_{kj}$'s into the Choi representation. At each iteration, we add an increment $\alpha \rightarrow \alpha+ \delta \alpha$, and optimize $\beta_{kj}$'s with the loss function being the negative minimal eigenvalue of $J(\Phi_{\alpha\boldsymbol{\beta}}^{(m)\dagger})$. When $\alpha^{(m)} < \alpha_{\text{max}}$, we can always get a set of optimized $\beta_{kj}$'s with the positive minimal eigenvalue of $J(\Phi_{\alpha\boldsymbol{\beta}}^{(m)\dagger})$. When $\alpha^{(m)} = \alpha_{\text{max}}$, we can get a set of optimized $\beta_{kj}$'s with the minimal eigenvalue of $J(\Phi_{\alpha\boldsymbol{\beta}}^{(m)\dagger})$ equal to $0$. And when $\alpha^{(m)} > \alpha_{\text{max}}$, we cannot find any set of optimized $\beta_{kj}$'s with the minimal eigenvalue of $J(\Phi_{\alpha\boldsymbol{\beta}}^{(m)\dagger}) $ greater than $0$. By this method we can determine the value of $\alpha_{\text{max}}$.

\subsection{Formulation in Semidefinite Programming Problem}
The task is to find an optimal $\Phi^{\dagger}$ is Eq.~\eqref{a_max_definition_app},
\begin{equation}
  \alpha_{\text{max}} = \{\max_{\boldsymbol{\beta}} \alpha: J(\Phi_{\alpha\boldsymbol{\beta}}^{\dagger}) \ge 0\},
\end{equation}
which can be expressed as a problem in semidefinite programming~\cite{gartner2012approximation}. 
The definition of a semidefinite programming in equational form is an optimization problem:
\begin{align}
    &\text{Maximize} \sum_{i,j=1}^{n}c_{ij}x_{ij},\nonumber\\
    &\text{subject to} \sum_{i,j=1}^{n}a_{ijk}x_{ij} = b_k,~k = 1,...,m, \nonumber\\
    &~~~~~~~~~~~~~~X\geq0,
\end{align}
where the $x_{ij}$ is the matrix element of a Hermitian matrix $X$, and $c_{ij}$, $a_{ijk}$ and $b_k$ are real coefficients.
And it can be written in a more compact form:
\begin{align}
    &\text{Maximize}~C\bullet X,\nonumber\\
    &\text{subject to}~A_k\bullet X=b_k,~k = 1,...,m,\nonumber\\
    &~~~~~~~~~~~~~~X\geq0,
\end{align}
where $C = (c_{ij})_{i,j=1}^n$, $A_k = (a_{ijk})_{i,j=1}^n$, and the notation $C\bullet X$ is defined as $C\bullet X = \sum_{i,j=1}^{n}c_{ij}x_{ij}$.

The optimization problem of Eq.~\eqref{a_max_definition_app} can be reformulated as 
\begin{align}\label{a_max_semidefinite_programming_app}
    &\text{Maximize}~\tr\big[J(\Phi_{\alpha\boldsymbol{\beta}}^{\dagger})\sum_{i=1}^2O_i\otimes Z_i\big]
    /\tr\big[(\sum_{i=1}^2O_i\otimes Z_i)^2\big],\nonumber\\
    &\text{subject to}~\tr\big[J(\Phi_{\alpha\boldsymbol{\beta}}^{\dagger})\big]=4^n,\nonumber\\
    &~~~~~~~~~~~~~~\tr\big[J(\Phi_{\alpha\boldsymbol{\beta}}^{\dagger})(I\otimes M_j)\big]=0,~j=1,...,4^n-1,\nonumber\\
    &~~~~~~~~~~~~~~\tr\big[J(\Phi_{\alpha\boldsymbol{\beta}}^{\dagger})(M_k\otimes I)\big]=0,~k=1,...,4^n-1,\nonumber\\
    &~~~~~~~~~~~~~~\tr\big[J(\Phi_{\alpha\boldsymbol{\beta}}^{\dagger})(\overline{O}_{1,i}\otimes Z_1)\big]=0,~i=1,...,4^n-1,\nonumber\\
    &~~~~~~~~~~~~~~\tr\big[J(\Phi_{\alpha\boldsymbol{\beta}}^{\dagger})(\overline{O}_{2,i}\otimes Z_2)\big]=0,~i=1,...,4^n-1,\nonumber\\
    &~~~~~~~~~~~~~~\tr\big[J(\Phi_{\alpha\boldsymbol{\beta}}^{\dagger})(a_1O_1\otimes Z_1-b_1O_2\otimes Z_2)\big]=0,\nonumber\\
    &~~~~~~~~~~~~~~J(\Phi_{\alpha\boldsymbol{\beta}}^{\dagger})\geq0,
\end{align}
where $O_i$ and $\overline{O}_{i,j}$ form a set of basis in the operator space, with $\tr(O_i\overline{O}_{i,j}) = 0$, for $j=1,...,4^n-1$ and $i=1,2$.
The coefficients $a_1$ and $b_1$ is obtained by solving the equation $\tr[(O_1\otimes Z_1+O_2\otimes Z_2)(a_1O_1\otimes Z_1-b_1O_2\otimes Z_2)] = 0$.

We can rewrite the Eq.~\eqref{a_max_semidefinite_programming_app} with coefficient matrix, 
\begin{align}\label{a_max_semidefinite_programming_matrix_app}
    &\text{Maximize}~C\bullet J(\Phi_{\alpha\boldsymbol{\beta}}^{\dagger}),\nonumber\\
    &\text{subject to}~I\bullet J(\Phi_{\alpha\boldsymbol{\beta}}^{\dagger})=4^n,\nonumber\\
    &~~~~~~~~~~~~~~(I\otimes M_j)^{T}\bullet J(\Phi_{\alpha\boldsymbol{\beta}}^{\dagger})=0,~j=1,...,4^n-1,\nonumber\\
    &~~~~~~~~~~~~~~(M_k\otimes I)^{T}\bullet J(\Phi_{\alpha\boldsymbol{\beta}}^{\dagger})=0,~k=1,...,4^n-1,\nonumber\\
    &~~~~~~~~~~~~~~(\overline{O}_{1,i}\otimes Z_1)^{T}\bullet J(\Phi_{\alpha\boldsymbol{\beta}}^{\dagger})=0,~i=1,...,4^n-1,\nonumber\\
    &~~~~~~~~~~~~~~(\overline{O}_{2,i}\otimes Z_2)^{T}\bullet J(\Phi_{\alpha\boldsymbol{\beta}}^{\dagger})=0,~i=1,...,4^n-1,\nonumber\\
    &~~~~~~~~~~~~~~(a_1O_1\otimes Z_1-b_1O_2\otimes Z_2)^{T}\bullet J(\Phi_{\alpha\boldsymbol{\beta}}^{\dagger})=0,\nonumber\\
    &~~~~~~~~~~~~~~J(\Phi_{\alpha\boldsymbol{\beta}}^{\dagger})\geq0,
\end{align}
with $C^{T} = \big(\sum_{i=1}^2O_i\otimes Z_i\big)/\tr\big[(\sum_{i=1}^2O_i\otimes Z_i)^2\big]$. 
Thus optimization problem of Eq.~\eqref{a_max_definition_app} can be expressed as a problem in semidefinite programming as Eq.~\eqref{a_max_semidefinite_programming_matrix_app}.

\section{Proof of Theorem 2}\label{app_proof_number}

In this section we give the proof of the Theorem 2 in the main text. 

Suppose we have obtained a well-trained optimized unital channel with $\alpha_{\text{max}}$, we will use it as a QNN to measure the expectation values of two 2-qubit observables $O_1$ and $O_2$ on a state $\rho$.
For direct projective measurements without QNNs, each measurement returns an eigenvalue, denoted as $\hat{o}_{1}$ for $O_1$ and $\hat{o}_{2}$ for $O_2$. In contrast, when applying a QNN, the Pauli-$Z$ measurements on the output qubits return outcomes $\hat{z}_{1},\hat{z}_{2} \in \{+1,-1\}$.

Define $N_O = N_{O_1} + N_{O_2}$ as the total number of copies of state $\rho$ prepared for projective measurement of $O_{1}$ and $O_{2}$, where $N_{O_j}$ is the sample size allocated to each observable. Let $N_Z$ denotes the number of copies of $\rho$ used for Pauli-$Z$ measurements with QNN. The estimate of the $j$-th observable expectation value is $\hat{O}_{j} = \sum_{i=1}^{N_{O_{j}}} \hat{o}_{j,i} /{N_{O_{j}}}$ for projective measurements, and $\hat{Z}_{j} = \sum_{i=1}^{N_{Z}} \hat{z}_{j,i} /(\alpha_{\text{max}}{N_{Z}})$ for Pauli $Z$ measurements with QNNs, where $\hat{o}_{j,i}$ ($\hat{z}_{j,i}$) is an eigenvalue of $O_j$ ($Z_j$) returned in the $i$-th projective measurement of $O_j$ ($Z_j$). When measurement numbers $N_O$ and $N_Z$ approach infinity, the limits of $\hat{Z}_j$ and $\hat{O}_j$ are denoted as $\mathbb{E}[\hat{z}_{j}/\alpha_{\text{max}}]$ and $\mathbb{E}[\hat{o}_{j}]$ respectively for $j \in \{1,2\}$. And $\mathbb{E}[\hat{z}_{j}/\alpha_{\text{max}}]=\mathbb{E}[\hat{o}_{j}] = \tr(\rho O_j)$. The variance of a random variable $\hat{o}$ is $\text{Var} [\hat{o}] = \mathbb{E}[\hat{o}^2] - \mathbb{E}[\hat{o}]^2$.
We will use the notation $\alpha$ instead of $\alpha_{\text{max}}$ in the following proof for simplicity.

\setcounter{theorem}{1}
\begin{theorem}\label{thm_complexity_app}
    (i) If we only focus on measuring one observable, the variance of Pauli-$Z$ measurement is always greater than direct projective measurements, 
    \begin{align}
        \text{Var} [\hat{z}_j/\alpha] \geq \text{Var} [\hat{o}_j],~j=1,2.
    \end{align}
    (ii) Suppose by preparing $N_O$ copies of $\rho$ and taking projective measurements, with high probability, it achieves
    \begin{align}
        \big|\hat{O}_{j}-\mathbb{E}[\hat{o}_j]\big|\leq \epsilon, ~j = 1,2.
    \end{align}
    Then, by preparing $N_Z = \mathcal{O}\big(\lambda N_O\big)$ copies of $\rho$, where
    \begin{align}\label{N_Z_without_Haar_app}
        \lambda = \frac{ 1  - \alpha^2 \min\{\mathbb{E}[\hat{o}_1]^2,\mathbb{E}[\hat{o}_2]^2\}}{\alpha^2[\text{Var}[\hat{o}_1] + \text{Var}[\hat{o}_2]]},
    \end{align}
    with high probability, the estimator $\hat{Z}_{j}$ achieves 
    \begin{align}
        \big|\hat{Z}_{j}-\mathbb{E}[\hat{o}_j]\big|\leq \epsilon, ~j = 1,2.
    \end{align}
   
   If $\rho = |\psi \rangle\langle\psi|$ is a pure state, the average number of copies under Haar measure becomes $N_Z = \mathcal{O}\big(\lambda_{H}N_O\big)$, 
   where 
   \begin{align}
       \lambda_{H} = \frac{d(d+1)-\alpha^2 \min \{\tr(O_1^2),\tr(O_2^2)\}}{d\alpha^2[\tr(O_1^2)+\tr(O_2^2)]},
   \end{align}
   and $d$ is the dimension of the Hilbert space.
\end{theorem}

\begin{proof}
    $(i)$ We evaluate the variance of $\hat{o}_1$ and $\hat{z}_1$.
    \begin{align}
    \text{Var} [\hat{o}_1] &= \mathbb{E}[\hat{o}_1^2] - \mathbb{E}[\hat{o}_1]^2 = 
    \tr(\rho O_1^2) -\tr(\rho O_1)^2.
    \end{align}
    \begin{align}
    \text{Var} [\hat{z}_1/\alpha] &= \frac{1}{\alpha^2}(\mathbb{E}[\hat{z}_1^2] - \mathbb{E}[\hat{z}_1]^2) \nonumber \\
            &= \frac{1}{\alpha^2}[(+1)^2 \langle 0|\rho_1|0\rangle + (-1)^2 \langle 1|\rho_1|1\rangle - \alpha^2(\rho O_1)^2] \nonumber \\
            &= 1/\alpha^2 - \tr(\rho O_1)^2,
\end{align}
where $\rho_1$ is the reduced density matrix on the first qubit, and $\mathbb{E}[\hat{z}_1/\alpha] = \mathbb{E}[\hat{o}_1]$.
Note that, $1/\alpha^2 \geq 1 \geq \tr (\rho O_1^2)$ is always satisfied for any $O_1$ with eigenvalues whose absolute values do not exceed $1$. Thus
\begin{align}\label{2_bit_information_inequality}
    \text{Var} [\hat{z}_1/\alpha]>\text{Var} [\hat{o}_1]
\end{align}
is always satisfied.

$(ii)$ Note that $\text{Var} [\hat{o}_1] \leq 2 \mathbb{E}[O_1] \leq 2$, and $|\hat{o}_1|^2\leq 1$. Use Bernstein inequality~\cite{Duchi} and we get
\begin{align}
    \text{Pr}\bigg[\frac{1}{N_{O_{1}}}\sum_{i=1}^{N_{O_{1}}} \hat{o}_{1,i} - \mathbb{E}[\hat{o}_1] \geq \varepsilon\bigg] \leq 
    2\exp [\frac{-N_{O_1}} {\varepsilon^2}{2(\mathbb{E}[\hat{o}_1^2] - \mathbb{E}[\hat{o}_1]^2) +\frac{4}{3}\mathbb{E}[\hat{o}_1] \varepsilon}], 
\end{align}
for $\delta\in (0,1)$, we get
\begin{align}
    N_{O_1} \geq \frac{2\ln (2/\delta)}{\varepsilon^2} \bigg[\mathbb{E}[\hat{o}_1^2] - \mathbb{E}[\hat{o}_1]^2 + \frac{2}{3} \mathbb{E}[\hat{o}_1] \varepsilon\bigg].
\end{align}

To measurement two observables $O_1$ and $O_2$, the total number of copies $\rho$ is $N_O = N_{O_1} + N_{O_2}$. Thus preparing 
\begin{align}\label{N_O}
    N_O \geq 
    \frac{2\ln (2/\delta)}{\varepsilon^2} \bigg[\mathbb{E}[\hat{o}_1^2] + \mathbb{E}[\hat{o}_2^2] - \mathbb{E}[\hat{o}_1]^2 -\mathbb{E}[\hat{o}_2]^2 + \frac{2}{3} (\mathbb{E}[\hat{o}_1] + \mathbb{E}[\hat{o}_2])\varepsilon\bigg]
\end{align}
copies implies
\begin{align}
        \big|\hat{O}_{j}-\tr(\rho O_j)\big|\leq \varepsilon, ~j = 1,2, 
\end{align}
$\text{with probability at least}~ 1-\delta$.

For the Pauli-$Z$ measurement with QNN, note that $|\hat{z}_1/\alpha - \mathbb{E}[\hat{o}_1]| \leq 1/\alpha + 1$, and $|\hat{z}_1|^2\leq 1$. Use Bernstein inequality we get
\begin{align}
    \text{Pr}\bigg[\frac{1}{\alpha N_{Z_{1}}}\sum_{i=1}^{N_{Z_{1}}} \hat{z}_{1,i} - \mathbb{E}[\hat{o}_1] \geq \varepsilon\bigg] \leq 
    2\exp [\frac{-N_{Z_1}} {\varepsilon^2}{2(\frac{1}{\alpha^2} - \mathbb{E}[\hat{o}_1]^2) +\frac{4}{3}\frac{\alpha+1}{\alpha} \varepsilon}], 
\end{align}
for $\delta\in (0,1)$, we get
\begin{align}
    N_{Z_1} \geq \frac{2\ln (2/\delta)}{\varepsilon^2} \bigg[\frac{1}{\alpha^2} - \mathbb{E}[\hat{o}_1]^2 +\frac{4}{3}\frac{\alpha+1}{\alpha} \varepsilon\bigg].
\end{align}

To measurement two observables $O_1$ and $O_2$ using QNN, the total number of copies $\rho$ is $N_Z = \max \{N_{Z_1}, N_{Z_2}\}$.
Thus preparing 
\begin{align}\label{N_Z}
    N_Z \geq \frac{2\ln (2/\delta)}{\varepsilon^2} \bigg[\frac{1}{\alpha^2} - \min\{\mathbb{E}[\hat{o}_1]^2,\mathbb{E}[\hat{o}_2]^2\} + \frac{2}{3}\frac{\alpha+1}{\alpha} \varepsilon\bigg]
\end{align}
copies implies
\begin{align}
        \big|\hat{Z}_{j}-\tr(\rho O_j)\big|\leq \varepsilon, ~j = 1,2, 
\end{align}
$\text{with probability at least}~ 1-\delta$.

We care about the asymptotic behavior of $\varepsilon$, Eq.~\eqref{N_O} and Eq.~\eqref{N_Z} become 
\begin{align}
    N_O \geq 
    \frac{2\ln (2/\delta)}{\varepsilon^2} \bigg[\mathbb{E}[\hat{o}_1^2] + \mathbb{E}[\hat{o}_2^2] - \mathbb{E}[\hat{o}_1]^2 -\mathbb{E}[\hat{o}_2]^2\bigg],
\end{align}
\begin{align}
    N_Z \geq \frac{2\ln (2/\delta)}{\varepsilon^2} \bigg[\frac{1}{\alpha^2} - \min\{\mathbb{E}[\hat{o}_1]^2,\mathbb{E}[\hat{o}_2]^2\}\bigg].
\end{align}

Thus we get the conclusion that, if we measure two observables $O_1$ and $O_2$ on state $\rho$ with projective measurements, to achieve precision $\varepsilon$ with probability $1-\delta$, we need to prepare 
\begin{align}
   N_O =\mathcal{O} \bigg(\frac{2\ln (2/\delta)}{\varepsilon^2} \bigg(\text{Var}[\hat{o}_1] + \text{Var}[\hat{o}_2]\bigg)\bigg).
\end{align}
If we measure two observables $O_1$ and $O_2$ on state $|\psi\rangle \langle\psi|$ with pauli-$Z$ measurements with QNNs, 
to achieve precision $\varepsilon$ with probability $1-\delta$, we need to prepare 
\begin{align}
   N_Z = \mathcal{O} \bigg(\frac{2\ln (2/\delta)}{\varepsilon^2} \big[\frac{1}{\alpha^2}  - \min\big\{\mathbb{E}[\hat{o}_1]^2,\mathbb{E}[\hat{o}_2]^2\big\}\big]\bigg) = \mathcal{O}\bigg(\frac{ 1  - \alpha^2 \min\{\mathbb{E}[\hat{o}_1]^2,\mathbb{E}[\hat{o}_2]^2\}}{\alpha^2[\text{Var}[\hat{o}_1] + \text{Var}[\hat{o}_2]]}N_O\bigg).
\end{align}

We next consider the average of pure state $ \rho = |\psi\rangle \langle\psi|$ under the Haar measure. By the formula of average of Haar measure on states~\cite{Mele2024introductiontohaar}, 
\begin{align}
    \mathbb{E}_{|\psi\rangle\sim \text{Haar}} \big[|\psi\rangle \langle\psi |\big] &= \frac{1}{d} \mathbb{I},\\
    \mathbb{E}_{|\psi\rangle\sim \text{Haar}} \big[|\psi\rangle \langle\psi |\otimes |\psi\rangle \langle\psi|\big] &= \frac{1}{d(d+1)} (\mathbb{I}+\mathbb{F}),
\end{align}
where $\mathbb{F}$ is the SWAP operator on the two tensor product Hilbert space.
Thus,
\begin{align}
     \mathbb{E}_{|\psi\rangle\sim \text{Haar}} \mathbb{E}[\hat{o}_i]^2 &=  \mathbb{E}_{|\psi\rangle\sim \text{Haar}} \tr(|\psi\rangle \langle\psi | O_i)^2\\
     &=\frac{1}{d(d+1)} \big[ \tr(O_i)^2 + \tr(O_i^2)\big]\\
     &=\frac{1}{d(d+1)} \big[ \tr(O_i^2)\big],\\
     \mathbb{E}_{|\psi\rangle\sim \text{Haar}} \mathbb{E}[\hat{o}_i^2] &=  \mathbb{E}_{|\psi\rangle\sim \text{Haar}} \tr(|\psi\rangle \langle\psi | O_i^2)\\
     &=\frac{1}{d} \big[\tr(O_i^2)\big],
\end{align}
and
\begin{align}
     \mathbb{E}_{|\psi\rangle\sim \text{Haar}} \text{Var} [\hat{o}_i] &= \mathbb{E}[\hat{o}_i^2] - \mathbb{E}[\hat{o}_i]^2\\
    &= \frac{1}{d} \big[\tr(O_i^2)\big] - \frac{1}{d(d+1)} \big[\tr(O_i^2)\big] \nonumber \\
    &=\frac{1}{d+1}\tr(O_i^2).
\end{align}
We get 
\begin{align}
   \mathbb{E}_{|\psi\rangle\sim \text{Haar}} N_O &\geq \mathbb{E}_{|\psi\rangle\sim \text{Haar}} \frac{2\ln (2/\delta)}{\varepsilon^2} \bigg[\mathbb{E}[\hat{o}_1^2] - \mathbb{E}[\hat{o}_1]^2 + \mathbb{E}[\hat{o}_2^2]  -\mathbb{E}[\hat{o}_2]^2\bigg]\nonumber \\
    &=\frac{2\ln (2/\delta)}{\varepsilon^2} \frac{1}{d(d+1)} [\tr(O_1^2)+\tr(O_2^2)],
\end{align}
and
\begin{align}
   \mathbb{E}_{|\psi\rangle\sim \text{Haar}} N_Z &\geq  \mathbb{E}_{|\psi\rangle\sim \text{Haar}}\frac{2\ln (2/\delta)}{\varepsilon^2} \bigg[\frac{1}{\alpha^2} - \min\big\{\mathbb{E}[\hat{o}_1]^2,\mathbb{E}[\hat{o}_2]^2\big\}\bigg]\nonumber\\
   &=\frac{2\ln (2/\delta)}{\varepsilon^2} \bigg[ \frac{1}{\alpha^2} - \frac{1}{d(d+1)}\min\big\{\tr(O_1^2),\tr(O_2^2)\big\}\bigg].
\end{align}

Thus, if we measure two observables $O_1$ and $O_2$ with projective measurements, to achieve precision $\varepsilon$ with probability $1-\delta$, we need to prepare 
\begin{align}
   \mathbb{E}_{|\psi\rangle\sim \text{Haar}} N_O =\mathcal{O} \bigg(\frac{2\ln (2/\delta)}{\varepsilon^2} \bigg(\frac{\tr(O_1^2)+\tr(O_2^2)}{d(d+1)} \bigg)\bigg).
\end{align}
If we measure two observables $O_1$ and $O_2$ with Pauli-$Z$ measurements with QNNs, 
to achieve precision $\varepsilon$ with probability $1-\delta$, we need to prepare 
\begin{align}
   \mathbb{E}_{|\psi\rangle\sim \text{Haar}} N_Z &=\mathcal{O} \bigg(\frac{2\ln (2/\delta)}{\varepsilon^2} \bigg[\frac{1}{\alpha^2} - \frac{1}{d(d+1)}\min\big\{\tr(O_1^2),\tr(O_2^2)\big\}\bigg]\bigg)\nonumber\\
   &= \mathcal{O}\bigg(\frac{d(d+1)-\alpha^2 \min \big\{\tr(O_1^2),\tr(O_2^2)\big\}}{d\alpha^2[\tr(O_1^2)+\tr(O_2^2)]}N_O\bigg).
\end{align}

\end{proof}


\section{Majorization Constraints for Identical Observables}
\label{app_majorization}
In this section we discuss the majorization limitation for the $\alpha_{\text{max}}$ if we have two identical observables $O_1 = O_2 = O$. For two $d$-dimensional vectors $\vec{a}$ and $\vec{b}$ whose components arranged in descending order, the majorization $\vec{a} \prec \vec{b}$ is defined as 
\begin{align}
    &(i)~\sum_{i=1}^{d'} a_i \leq \sum_{i=1}^{d'} b_i,~1\leq d' \leq d,\nonumber\\
    \text{and}~&(ii)~\sum_{i=1}^d a_i = \sum_{i=1}^d b_i = \text{constant}. 
\end{align}

We note that mixed-unitary channel is related to the Uhlmann theorem~\cite{nielsen2002introduction}, which states that there exists a mixed-unitary channel $\mathcal{E}$ such that $\mathcal{E}(A) = B$ if and only if $A\succ B$. And $A\succ B$ if and only if $\lambda_{A} \succ \lambda_{B}$, where $\lambda_{O}$ is the vector of eigenvalues for the operator $O$ in descending order. For two operators case $\alpha O_1 = \mathcal{E}(Z_1)$ and $\alpha O_2 = \mathcal{E}(Z_2)$, we find that three limitations need to be satisfied. The first is $Z_1 \succ \alpha O_1 $. The second is $Z_2 \succ \alpha O_2 $. And the third is $xZ_1 + yZ_2 \succ \alpha( xO_1 + yO_2) $, with $x$ and $y$ are any real numbers. The third limitation is obtained from the linearity of the quantum channel: $\alpha O_1 = \mathcal{E}(Z_1)$ and $\alpha O_2 = \mathcal{E}(Z_2)$ leads to $\alpha (xO_1 +yO_2)= \mathcal{E}(xZ_1+yZ_2)$. 

If we choose $O_1 = O_2 = O$, the third limitation above gives 
\begin{align}
    Z_1 + Z_2 \succ 2\alpha O.
\end{align}
Here we set $x=y=1$ because it gives the tightest limitation for $\alpha$.
And the limitation for eigenvalues is 
\begin{align}
    (2,0,0,-2) \succ (2\alpha o_1, 2\alpha o_2, 2\alpha o_3, 2\alpha o_4),
\end{align}
with $o_i$ is the $i$-th eigenvalue of $O$.
By the definition, it becomes
\begin{align}
    2\geq2\alpha o_1,~~2\geq2\alpha (o_1+o_2),~~2\geq2\alpha (o_1+o_2+o_3).
\end{align}
Thus
\begin{align}
    \alpha\leq \frac{1}{o_1},~~\alpha\leq \frac{1}{o_1+o_2},~~\alpha\leq \frac{1}{o_1+o_2+o_3}.
\end{align}
If we want a large $\alpha$, then $o_i$'s need to be small. However, if we want a large $\tr(O^2)$, then then $o_i$'s need to be large. 
Thus, the requirement for $\lambda_{H}<1$, which is $\alpha_{\max}^{2}\tr(O^2)> (d+1)d/(2d+1)>2$, can not be realized for general observables.

\end{document}